\documentclass[12pt,a4paper,american]{article}

\usepackage{babel}
\usepackage{amssymb}
\usepackage{amsfonts}
\usepackage{amsthm}
\usepackage{mathrsfs}

\usepackage{setspace}
\usepackage[ref,nosort]{overcite}

\makeatletter

\renewcommand\@biblabel[1]{$^{#1}$}

\renewcommand{\@seccntformat}[1]{\csname the#1\endcsname.\hspace{0.5em}}

\renewenvironment{abstract}{%
  \begin{center}%
    {\bfseries \abstractname\vspace{-.5em}\vspace{\z@}}%
  \end{center}%
  \quotation}

\newcommand{\lyxaddress}[1]{
\par {\raggedright #1
\vspace{1.4em}
\noindent\par}
}

\makeatother

\newcommand{\sH}{\mathscr{H}}
\newcommand{\sL}{\mathscr{L}}
\newcommand{\sI}{\mathscr{I}}
\newcommand{\sF}{\mathscr{F}}
\newcommand{\sD}{\mathscr{D}}
\newcommand{\sB}{\mathscr{B}}
\newcommand{\R}{\mathbb{R}}
\newcommand{\KK}{\mathcal{K}}
\newcommand{\fg}{\mathfrak{g}}
\newcommand{\fh}{\mathfrak{h}}

\newtheorem{thm}{Theorem}
\newtheorem{prop}[thm]{Proposition}
\newtheorem{lem}[thm]{Lemma}
\newtheorem{cor}[thm]{Corollary}

\theoremstyle{remark}
\newtheorem{rem}[thm]{Remark}
\newtheorem*{rem*}{Remark}
\newtheorem*{rems*}{Remarks}

\theoremstyle{definition}
\newtheorem{define}[thm]{Definition}
\newtheorem*{def*}{Definition}

\newcommand{\Ran}{\mathop\mathrm{Ran}\nolimits}
\newcommand{\supp}{\mathop\mathrm{supp}\nolimits}
\newcommand{\Tr}{\mathop\mathrm{Tr}\nolimits}
\newcommand{\id}{\mathop\mathrm{id}\nolimits}

\newcommand{\aall}{a.a.\ }
\newcommand{\aevry}{a.e.\ }

\begin{document}

\title{\vspace{-1\baselineskip}
  Generalized Bloch analysis and propagators on Riemannian manifolds
  with a discrete symmetry}

\author{P.~Koc\'abov\'a, P.~\v{S}\v{t}ov\'\i\v{c}ek}
\date{}

\maketitle

\lyxaddress{Department of Mathematics, Faculty of Nuclear Science,
  Czech Technical University, Trojanova 13, 120 00 Prague, Czech
  Republic}

\begin{abstract}
  \noindent
  We consider an invariant quantum Hamiltonian $H=-\Delta_{LB}+V$ in
  the $L^{2}$ space based on a Riemannian manifold $\tilde{M}$ with a
  countable discrete symmetry group $\Gamma$. Typically, $\tilde{M}$
  is the universal covering space of a multiply connected Riemannian
  manifold $M$ and $\Gamma$ is the fundamental group of $M$. On the
  one hand, following the basic step of the Bloch analysis, one
  decomposes the $L^{2}$ space over $\tilde{M}$ into a direct integral
  of Hilbert spaces formed by equivariant functions on $\tilde{M}$.
  The Hamiltonian $H$ decomposes correspondingly, with each component
  $H_{\Lambda}$ being defined by a quasi-periodic boundary condition.
  The quasi-periodic boundary conditions are in turn determined by
  irreducible unitary representations $\Lambda$ of $\Gamma$. On the
  other hand, fixing a quasi-periodic boundary condition (i.e., a
  unitary representation $\Lambda$ of $\Gamma$) one can express the
  corresponding propagator in terms of the propagator associated to
  the Hamiltonian $H$. We discuss these procedures in detail and show
  that in a sense they are mutually inverse.
\end{abstract}

\newpage

\section{Introduction}

It has been demonstrated in Ref.~\citen{aharonovbohm} that in the
formalism of quantum mechanics vector potentials are more significant
than in classical mechanics. Geometrically this observation consists
of the fact that over a connected but multiply connected manifold $M$
there exist non-equivalent vector bundles with connections (covariant
derivatives) whose curvatures are equal. Here we are going to focus
exclusively on the case when the curvature vanishes. It is not
difficult to see that all flat Hermitian connections on $M$, up to
equivalence, are in one-to-one correspondence with unitary
representations of the fundamental group $\Gamma=\pi_{1}(M)$. Given a
flat Hermitian connection the corresponding unitary representation of
$\Gamma$ is defined by the parallel transport along closed paths with
a fixed base point. In the physical terminology, the parallel
transport along a closed path produces nothing but the nonintegrable
phase factor \cite{wuyang,horvathy}. Conversely, denote by $\tilde{M}$
the universal covering space of $M$. Then $\tilde{M}\to M$ is a
principal fiber bundle with the structure group $\Gamma$, and since
$\dim\tilde{M}=\dim M$ there exists exactly one connection in this
bundle which is necessarily flat. Given a unitary representation
$\Lambda$ of $\Gamma$ one can associate to this principal fiber bundle
a Hermitian vector bundle with a connection which is again flat
\cite{kobayashinomizu}.

On the physical level this means that there may exist, depending on
$\pi_{1}(M)$, non-equivalent quantum Hamiltonians describing a
particle moving on $M$ in a vanishing gauge field. In the physical
literature there was described a construction of the propagator
associated to such a Hamiltonian. The construction is based on the
notion of the Feynman path integral and assumes the knowledge of the
free propagator on the universal covering space $\tilde{M}$. The
propagator on $M$ connecting points $x$ and $y$ is then constructed as
a weighted sum running over all homotopy classes of paths from $x$ to
$y$, the summands are expressed in terms of the free propagator on
$\tilde{M}$ and the weights depend on a given representation $\Lambda$
of $\Gamma$. In Ref.~\citen{schulman0} one treats the case of a
circle, and in Ref.~\citen{schulman1} the example of the Aharonov-Bohm
effect is worked out. All this material can be also found in
Ref.~\citen{schulman2}.  Generalizations to non-Abelian gauge groups
are discussed in Refs.~\citen{sundrumtassie,ohetal}. Let us also note
that in the mathematical literature an analogues formula is known to
be valid, under certain assumptions, for heat
kernels\cite{atiyah,sunada2}.

Though the formula for the propagator is derived in
Ref.~\citen{schulman1} rather formally it turns out to be quite
effective even when considering more complicated topologies. As a
distinguished example of this kind one may point out two-dimensional
quantum systems describing the Aharonov-Bohm effect with two
solenoids, possibly with an additional scalar potential
\cite{mashkevichetal}. The propagator for such a system is expressed
in an explicit manner as an infinite series in Ref.~\citen{pla89}, and
more details on the method and computations are provided in
Ref.~\citen{ksexample}. The same example with a spin is discussed in
Ref.~\citen{geylerps}. Moreover, this formula makes it possible to
treat the scattering problem in this model as well.
%\cite{pla91,pra93,duke94}
%%
\hskip -0.2em
$^{\!\textrm{\scriptsize\citen{pla91,pra93},\,\citen{duke94}}}$
\hskip 0.2em
The scattering amplitude has been also derived in
Ref.~\citen{itotamura} by completely different technical tools, see
also Refs.~\citen{itotamura2,tamura}. As one can verify by a
straightforward computation, this is an encouraging observation that
both methods finally lead to the same result (compare formula (12) in
Ref.~\citen{pra93} to Theorem~1.1 in Ref.~\citen{itotamura}\,). Let us
note that some other two-dimensional systems with similar features
have been discussed in the literature recently
\cite{hannaythain,giraudetal}.

The formula for the propagator on multiply connected spaces, as
derived in Ref.~\citen{schulman1} in the framework of the Feynman path
integral, is the central topic of the current paper. We shall call it
loosely the Schulman's ansatz. Our goal is to find a mathematically
rigorous interpretation of this formula. In a more general setting, we
shall consider the situation when $\tilde{M}$ is a manifold with a
discrete symmetry group $\Gamma$ and $M=\tilde{M}/\Gamma$. Thus
$\tilde{M}$ is a covering space of $M$ though not necessarily
universal. Treating this problem we realized quickly that one has to
consider the construction also from the opposite side. By the
Schulman's ansatz one relates to the free propagator on $\tilde{M}$
and to any unitary representation $\Lambda$ of $\Gamma$ a propagator
corresponding to a Hamiltonian on $M$. Conversely, one may ask how to
reconstruct the free propagator on $\tilde{M}$ from the knowledge of
the family of propagators on $M$ when $\Lambda$ runs over all unitary
irreducible representations (up to equivalence) of $\Gamma$. The
inverse procedure leads to a generalization of the Bloch decomposition
which we are going to discuss as well.

A generalization of the Bloch analysis has been proposed by Sunada
\cite{sunada}. This method has been used systematically to reveal the
band structure of spectrum for a $\Gamma$-periodic elliptic operator
on a non-compact Riemannian manifold $\tilde{M}$ with a discrete
symmetry group $\Gamma$ under the assumption that the quotient
$\tilde{M}/\Gamma$ is compact. For recent progress in this direction
see also Refs.~\citen{lledopostI,lledopostII}. Further, in
Ref.~\citen{aschetal} one considers magnetic Hamiltonians on the
two-dimensional torus in the case when the magnetic field is
determined by an integral two-form (or, in other words, when the
magnetic flux trough 2-cycles is quantized in accordance with the
famous Dirac's rule for magnetic monopoles). It is shown there that
the Bochner Laplacian over the universal covering of the torus (the
plane) decomposes into a direct integral whose components are all
equivalence classes of Bochner Laplacians over the torus with the
prescribed magnetic field. This result is then extended in
Ref.~\citen{gruber1} to more general compact Riemannian manifolds $M$
whose fundamental group $\Gamma$ is Abelian. Again one constructs a
direct integral over the Pontryagin dual $\hat{\Gamma}$. A more
algebraic approach in the spirit of the Gelfand's representation
theorem is developed in Ref.~\citen{gruber2}. Given a discrete
symmetry group $\Gamma$ one does not work directly with the dual space
$\hat{\Gamma}$ but instead with a $C^{\ast}$-algebra of continuous
functions on that space. Thus the basic topic in Ref.~\citen{gruber2}
is the spectral analysis of elliptic operators on Hilbert
$C^{\ast}$-modules over noncommutative $C^{\ast}$-algebras. The
symmetry group $\Gamma$ can be noncommutative, its action is assumed
to be co-compact.

In the current paper we recall the Bloch decomposition in a form which
is rather close to that due to Sunada \cite{sunada}. For a more
detailed discussion see Remark~\ref{rem:comparison} below in the text.
Let us stress however that our motivation is quite different than are
those in the above cited papers and thus we do not aim at all at the
spectral analysis of elliptic operators in question. For our purposes
we need only to construct a decomposition of a Hamiltonian on the
covering space $\tilde{M}$ into a direct integral over the dual space
$\hat{\Gamma}$ whose components are Hamiltonians based on the manifold
$M=\tilde{M}/\Gamma$.

Let us point out some other features of our approach. Firstly, here we
are interested only in the case when the gauge field vanishes, and
thus, in a suitable formalism, we deal just with Laplace-Beltrami
operators $\Delta_{LB}$. Furthermore, rather than working with $L^{2}$
spaces of sections in vector bundles over $M$ we prefer to work with
spaces of equivariant functions on the covering space $\tilde{M}$
since this way one avoids the explicit use of vector potentials.
Secondly, we do not require the manifold $M$ to be compact. For this
more general setting we have to distinguish the Friedrichs extension
as the preferred self-adjoint extension of a semibounded symmetric
operator $-\Delta_{LB}+V$ defined on smooth functions with compact
supports. Finally, we restrict ourselves to type I discrete groups
$\Gamma$ since for these groups the generalized harmonic analysis is
well established \cite{shtern}. Unfortunately, with this restriction,
some interesting physical models are not covered. Nevertheless one may
hope that further generalizations of our approach are possible since
the generalized Fourier analysis has been developed also for some
other groups which are not included among the type I groups.

The paper is organized as follows. In Section~2 we introduce some
basic notation and notions, and we specify more precisely the goals of
the paper. Several facts concerning the generalized Fourier analysis
that are important for our approach are recalled in Section~3. In
Section~4 we discuss the generalized Bloch decomposition. Further we
concentrate on the propagator which is regarded as the generalized
kernel of the unitary evolution operator. Thus, to deal with the
propagator, we apply the Schwartz kernel theorem. In addition, we have
to adjust this theorem to our purposes. This is done in Section~5.
Finally, Section~6 is devoted to the desired interpretation of the
Schulman's ansatz.

\section{Formulation of the problem}

We consider a smooth connected Riemannian manifold $\tilde{M}$ (which
is supposed to be Hausdorff and second countable). Denote by
$\tilde{\mu}$ the measure on $\tilde{M}$ induced by the Riemannian
metric. All $L^{p}$ spaces based on $\tilde{M}$ will be understood
with this measure.  Furthermore, there is given an at most countable
discrete group $\Gamma$ acting on $\tilde{M}$ as a symmetry group,
i.e., the Riemannian metric is $\Gamma$-invariant. The action is
assumed to be smooth, free and proper (or, by another frequently used
terminology, properly discontinuous). Let us recall that under these
assumptions any element $s\in\Gamma$ different from the unity has no
fixed points on $\tilde{M}$, and for any compact set
$K\subset\tilde{M}$, the intersection $K\cap s\cdot K$ is nonempty
only for finitely many elements $s\in\Gamma$. This also implies that
any point $y\in\tilde{M}$ has a neighborhood $U$ such that the sets
$s\cdot U$, $s\in\Gamma$, are mutually disjoint
\cite[Corollary~12.10]{lee}.

The quotient $M=\tilde{M}/\Gamma$ is again a connected Riemannian
manifold \cite[Proposition~4.1.23]{abrahammarsden}. Let $\mu$ denote
the induced measure on $M$. Again, all $L^{p}$ spaces based on $M$ are
understood with this measure. Thus we get a principal fiber bundle
$\pi:\tilde{M}\to{}M$ with the structure group $\Gamma$.

In some applications the following example is of interest. One starts
from a connected Riemannian manifold $M$. Let $\tilde{M}$ be the
universal covering space of $M$ and $\Gamma=\pi_{1}(M)$ be the
fundamental group of $M$. Then $\Gamma$ is at most countable
\cite[Theorem 8.11]{lee}, $\Gamma$ acts on $\tilde{M}$ smoothly,
freely and properly \cite[Chapter~21]{kosniowski}, and one can
naturally identify $M$ with $\tilde{M}/\Gamma$ (see also
Proposition~5.9 in Ref.~\citen{kobayashinomizu}).

Let us denote by $L_{s}$ the left action of $s\in\Gamma$ on
$\tilde{M}$, i.e., $L_{s}(y)=s\cdot y$ for $y\in\tilde{M}$. Given a
unitary representation $\Lambda$ of $\Gamma$ on a separable Hilbert
space $\sL_{\Lambda}$ one constructs the Hilbert space $\sH_{\Lambda}$
formed by $\Lambda$-equivariant vector-valued functions on $\tilde{M}$
(more precisely, by their equivalence classes modulo measure zero). In
more detail, any function $\psi\in\sH_{\Lambda}$ is measurable with
values in $\sL_{\Lambda}$ and
\[
\forall s\in\Gamma,\,
L_{s}^{\ast}\psi=\Lambda(s)\psi.
\]
Furthermore, the norm of $\psi$ induced by the scalar product, as
introduced below, must be finite.  Here, as usual,
$f^{*}:\mathrm{Fun}(N)\to\mathrm{Fun}(M)$ is the pull-back mapping
associated to $f:M\to N$ and some appropriate function spaces
$\mathrm{Fun}(M)$, $\mathrm{Fun}(N)$ based on the sets $M$ and $N$,
respectively. In order to keep the notation simple the same symbol,
$f^{\ast}$, will be used independently of the concrete nature of
functional spaces in question (they may be formed, for example, by
smooth functions or square integrable functions). The scalar product
of $\psi_{1},\psi_{2}\in\sH_{\Lambda}$ is defined by
\begin{equation}
  \langle\psi_{1},\psi_{2}\rangle
  = \int_{M}\langle\psi_{1}(y),\psi_{2}(y)\rangle\,d\mu(x).
  \label{eq:sprod_psi12}
\end{equation}

\begin{rem}
  \label{rem:int_f_invar}
  In (\ref{eq:sprod_psi12}) and everywhere in what follows we use the
  following convention. If $f$ is a measurable function on $\tilde{M}$
  such that $f$ is constant on the fibers of $\pi$ (equivalently, $f$
  is $\Gamma$-invariant) then $f=\pi^{\ast}g$ for some, essentially
  unique, measurable function $g$ on $M$. If $g\in L^{1}(M)$ then by
  the integral
  \[
  \int_{M}f(y)\, d\mu(x)
  \]
  we mean $\int_{M}g\, d\mu$. Notice that, with this convention, if
  $f\in L^{1}(\tilde{M})$ is arbitrary then
  \begin{equation}
    \int_{\tilde{M}}f\,\mathrm{d}\tilde{\mu}
    = \int_{M}\sum_{s\in\Gamma}L_{s}^{*}f(y)\,\mathrm{d}\mu(x).
    \label{eq:int_Mtild_M}
  \end{equation}
  One can easily see that the sum $\sum_{s\in\Gamma}L_{s}^{*}f(y)$
  converges absolutely almost everywhere on $\tilde{M}$.
\end{rem}

Let $\Delta_{LB}$ be the Laplace-Beltrami operator on $\tilde{M}$. As
a differential operator on $C_{0}^{\infty}(\tilde{M})$, $\Delta_{LB}$
is unambiguously determined by the equality
\begin{displaymath}
  \forall\varphi_1,\varphi_2\in C_{0}^{\infty}(\tilde{M}),\quad
  \langle\varphi_1,-\Delta_{LB}\varphi_2\rangle
  = \int_{\tilde{M}}\tilde{\mathfrak{g}}
  (\mathrm{d}\varphi_1,\mathrm{d}\varphi_2)\,\mathrm{d}\tilde{\mu}
\end{displaymath}
where the scalar product on the LHS is understood in $L^2(\tilde{M})$
and $\tilde{\mathfrak{g}}$ is the Riemannian metric defined on the
cotangent spaces on $\tilde{M}$. Suppose further that there is given a
measurable $\Gamma$-invariant bounded real function $V(y)$ on
$\tilde{M}$. Then one can introduce with the aid of the Friedrichs
extension the Hamiltonian $H=-\Delta_{LB}+V$ as a selfadjoint operator
on $L^{2}(\tilde{M})$. Actually, the differential operator
$-\Delta_{LB}+V$ is symmetric and bounded below on the domain
$C_{0}^{\infty}(\tilde{M})$. Moreover, the invariance of the
Riemannian metric and the invariance of $V$ imply that $H$ is a
$\Gamma$-invariant operator (i.e., $H$ commutes with all
$L_{s}^{\ast}$, $s\in\Gamma$).

To the same differential operator, $-\Delta_{LB}+V$, one can relate a
selfadjoint operator $H_{\Lambda}$ on the space $\sH_{\Lambda}$ for
any unitary representation $\Lambda$ of $\Gamma$. First one constructs
a linear subspace in $\sH_{\Lambda}$ formed by smooth vector-valued
functions. Let us define
\begin{equation}
  \Phi_{\Lambda} := \sum_{s\in\Gamma}L_{s}^{\ast}
  \otimes\Lambda(s^{-1}):C_{0}^{\infty}(\tilde{M})\otimes\sL_{\Lambda}
  \to\sH_{\Lambda}
  \label{eq:def_PhiL}
\end{equation}
If $\varphi\in C_{0}^{\infty}(\tilde{M})$, $v\in\sL_{\Lambda}$, then
\begin{equation}
  \left(\Phi_{\Lambda}\varphi\otimes v\right)(y)
  = \sum_{s\in\Gamma}\varphi(s\cdot y)\,\Lambda(s^{-1})v,
  \label{eq:PhiL_phiv}
\end{equation}
and on any compact set $K\subset\tilde{M}$, only a finite number of
summands on the RHS of (\ref{eq:PhiL_phiv}) do not vanish (for the
action of $\Gamma$ is proper). Consequently, the vector-valued
function $\Phi_{\Lambda}\,\varphi\otimes{}v$ is smooth.

From definition (\ref{eq:def_PhiL}) one immediately finds that
\begin{eqnarray}
  \forall s\in\Gamma, & & L_{s}^{\ast}\circ\Phi_{\Lambda}
  = \Lambda(s)\Phi_{\Lambda},\label{eq:LsPhiL}\\
  \forall s\in\Gamma, & & \Phi_{\Lambda}\circ(L_{s}^{\ast}\otimes1)
  =\Phi_{\Lambda}\circ(1\otimes\Lambda(s)).
  \label{eq:PhiLLs}
\end{eqnarray}
Property (\ref{eq:LsPhiL}) implies that
$\Phi_{\Lambda}\,\varphi\otimes{}v$ is actually $\Lambda$-equivariant.
It is also not difficult to see from (\ref{eq:PhiL_phiv}) that the
norm of $\Phi_{\Lambda}\,\varphi\otimes{}v$ in the Hilbert space
$\sH_{\Lambda}$ is finite.

Furthermore, using (\ref{eq:int_Mtild_M}) and some simple
manipulations one finds that for any
$\varphi\in{}C_{0}^{\infty}(\tilde{M})$ and $v\in\sL_{\Lambda}$,
\begin{equation}
  \forall\psi\in\sH_{\Lambda},\,
  \left\langle \psi,\Phi_{\Lambda}\varphi
    \otimes v\right\rangle
  = \int_{\tilde{M}}\varphi(y)\langle\psi(y),v\rangle\,
  \mathrm{d}\tilde{\mu}(y).
  \label{eq:prod_psi_PhiL}
\end{equation}
From this one deduces that the range of $\Phi_{\Lambda}$ is dense in
$\sH_{\Lambda}$. As a particular case of (\ref{eq:prod_psi_PhiL}) we
have
\begin{equation}
  \left\langle \Phi_{\Lambda}\varphi_{1}
    \otimes v_{1},\Phi_{\Lambda}\varphi_{2}\otimes v_{2}\right\rangle
  = \sum_{s\in\Gamma}
  \left\langle\Lambda(s^{-1})v_{1},v_{2}\right\rangle
  \left\langle L_{s}^{*}\varphi_{1},\varphi_{2}\right\rangle.
  \label{eq:prod_PhiL_PhiL}
\end{equation}
Here the scalar product
$\left\langle{}L_{s}^{*}\varphi_{1},\varphi_{2}\right\rangle$ is
understood in $L^{2}(\tilde{M})$ and, once again, it is nonzero only
for finite number of elements $s\in\Gamma$ for the action of $\Gamma$
is proper.

The Laplace-Beltrami operator is well defined on
$\Ran(\Phi_{\Lambda})$. Since $\Delta_{LB}$ commutes with $L_{s}^{*}$,
$s\in\Gamma$, one has
\begin{equation}
  \Delta_{LB}\Phi_{\Lambda}[\varphi\otimes v]
  = \Phi_{\Lambda}[\Delta_{LB}\varphi\otimes v].
  \label{eq:DPhiL_PhiLD}
\end{equation}

\begin{lem}
  The differential operator $-\Delta_{LB}$ is positive
  on the domain $\Ran(\Phi_{\Lambda})\subset\sH_{\Lambda}$.\\
\end{lem}

\begin{proof}
  Denote by $\mathfrak{g}$ the Riemannian metric defined on the
  cotangent spaces on $M$. Recall that $\tilde{\mathfrak{g}}$ has a
  similar meaning on the Riemannian manifold $\tilde{M}$. Using
  (\ref{eq:prod_PhiL_PhiL}) and equality (\ref{eq:int_Mtild_M}) in
  Remark~\ref{rem:int_f_invar} one derives that
  \begin{eqnarray}
    & & \hspace{-2em}
    \left\langle \Phi_{\Lambda}\varphi_{1}
      \otimes v_{1},-\Delta_{LB}\Phi_{\Lambda}\varphi_{2}
      \otimes v_{2}\right\rangle
    \,=\, \sum_{s\in\Gamma}\langle v_{1},\Lambda(s^{-1})v_{2}\rangle
    \int_{\tilde{M}}\tilde{\mathfrak{g}}
    (\mathrm{d}\overline{\varphi_{1}},L_{s}^{*}\mathrm{d}
    \varphi_{2})\,\mathrm{d}\tilde{\mu}\nonumber
    \\
    & & \hspace{7em}
    =\, \int_{M}\sum_{t\in\Gamma}\sum_{s\in\Gamma}
    \langle\Lambda(t^{-1})v_{1},\Lambda(s^{-1})v_{2}\rangle\,
    \tilde{\mathfrak{g}}(L_{t}^{*}\mathrm{d}
    \overline{\varphi_{1}}(y),L_{s}^{*}\mathrm{d}\varphi_{2}(y))\,
    \mathrm{d}\mu(x).\nonumber \\
    & & \mbox{}
    \label{eq:prod_PhiL_PhiL2}
  \end{eqnarray}
  Choose $\psi\in\Ran\Phi_{\Lambda}$. Then $K=\pi(\supp\psi)$ is a
  compact subset of $M$. There exists a finite open covering of $K$,
  $K\subset\bigcup_{i=1}^{n}U_{i}$, such that the fiber bundle
  $\pi:\tilde{M}\to M$ is trivial over each $U_{i}$, i.e., there exist
  smooth sections $\sigma_{i}:U_{i}\to\tilde{M}$,
  $\pi\circ\sigma_{i}=\id_{U_{i}}$. Furthermore, for this covering
  there exists a partition of unity
  $\left\{\eta_{i}\right\}_{i=1}^{n}$ where
  $\eta_{i}\in{}C_{0}^{\infty}(U_{i})$, $0\leq\eta_{i}\leq1$, and
  $\sum\eta_{i}\equiv1$ on a neighborhood of $K$. Choose also
  auxiliary functions $\xi_{i}\in C_{0}^{\infty}(U_{i})$ so that
  $0\leq\xi_{i}\leq1$ and $\xi_{i}\equiv1$ on a neighborhood of
  $\supp\eta_{i}$. Set
  \[
  \tilde{\xi}_{i} = \pi^{*}\xi_{i}\in C^{\infty}(\pi^{-1}(U_{i})).
  \]
  The only reason of introducing the functions $\tilde{\xi}_{i}$ is to
  cope with the fact, in the formulas below, that the sections
  $\sigma_{i}$ are defined only locally. One simply uses the natural
  inclusion $C_{0}^{\infty}(U_{i})\subset C_{0}^{\infty}(M)$. If
  $\varphi_{1},\varphi_{2}\in C_{0}^{\infty}(\tilde{M})$ satisfy
  $\supp\varphi_{1},\supp\varphi_{2}\subset\pi^{-1}(K)$ then
  expression (\ref{eq:prod_PhiL_PhiL2}) equals
  \[
  \sum_{i=1}^{n}\sum_{t\in\Gamma}\sum_{s\in\Gamma}
  \langle\Lambda(t^{-1})v_{1},\Lambda(s^{-1})v_{2}\rangle
  \int_{M}\eta_{i}\,\mathfrak{g}(\sigma_{i}^{*}
  \mathrm{d}(\tilde{\xi}_{i}L_{t}^{*}
  \overline{\varphi_{1}}),\sigma_{i}^{*}\mathrm{d}
  (\tilde{\xi}_{i}L_{s}^{*}\varphi_{2}))\,\mathrm{d}\mu.
  \]
  Let $\mathfrak{s}_{i}$ be a positive sesquilinear form on
  $C_{0}^{\infty}(M)\otimes\sL_{\Lambda}$ defined by
  \[
  \mathfrak{s}_{i}(\zeta_{i}\otimes v_{1},\zeta_{2}
  \otimes v_{2}) = \langle v_{1},v_{2}\rangle
  \int_{M}\eta_{i}\,\mathfrak{g}
  (\mathrm{d}\overline{\zeta_{1}},\mathrm{d}\zeta_{2})\,\mathrm{d}\mu.
  \]
  With this notation we have
  \[
  \left\langle\Phi_{\Lambda}\varphi_{1}
    \otimes v_{1},-\Delta_{LB}\Phi_{\Lambda}\varphi_{2}
    \otimes v_{2}\right\rangle
  = \sum_{i=1}^{n}\mathfrak{s}_{i}(\sigma_{i}^{*}
  (\tilde{\xi}_{i}\Phi_{\Lambda}\varphi_{1}
  \otimes v_{1}),\sigma_{i}^{*}
  (\tilde{\xi}_{i}\Phi_{\Lambda}\varphi_{2}\otimes v_{2})).
  \]

  Thus we arrive at the following conclusion. Set
  $\phi_{i}=\sigma_{i}^{*}(\tilde{\xi}_{i}\psi)\in{}C_{0}^{\infty}(M)
  \otimes\sL_{\Lambda}$. Then
  \[
  \left\langle \psi,-\Delta_{LB}\psi\right\rangle
  = \sum_{i=1}^{n}\mathfrak{s}_{i}(\phi_{i},\phi_{i})\geq0.
  \]
  This completes the verification.
\end{proof}

Clearly, since the function $V(y)$ is $\Gamma$-invariant the
multiplication operator by $V$ is well defined in the Hilbert space
$\sH_{\Lambda}$. Now the definition of the Hamiltonian $H_{\Lambda}$
is straightforward. This is the Friedrichs extension of the
differential operator $-\Delta_{LB}+V$ considered on the domain
$\Ran\Phi_{\Lambda}$.

Let us denote by $U(t)=\exp(-itH)$, $t\in\mathbb{R}$, the evolution
operator in $L^{2}(\tilde{M})$. Similarly,
$U_{\Lambda}(t)=\exp(-itH_{\Lambda})$, $t\in\mathbb{R}$, is the
evolution operator in $\sH_{\Lambda}$ where $\Lambda$ is a unitary
representation of $\Gamma$. Denote by $\hat{\Gamma}$ the dual space to
$\Gamma$ (the quotient space of the space of irreducible unitary
representations of $\Gamma$). In the current paper we wish to address
the following two mutually complementary problems. First, to express
$U(t)$ in terms of $U_{\Lambda}(t)$, $\Lambda\in\hat{\Gamma}$.
Second, to express $U_{\Lambda}(t)$ in terms of $U(t)$ for a fixed
unitary representation $\Lambda$ of $\Gamma$. It turns out that
answers to both problems do exist. A solution to the former one is
provided by the generalized Bloch decomposition. A solution to the
latter problem is given by a formula known from the theoretical
physics (the Schulman's ansatz) \cite{schulman1,schulman2}.

\section{The generalized Fourier analysis}

The generalized harmonic analysis is well established for locally
compact groups of type~I \cite{shtern}. This is why we restrict
ourselves to the case when $\Gamma$ is a type~I group. Countable
discrete groups of type~I are well characterized \cite[Satz~6]{thoma}.

\begin{thm}[Thoma]
  A countable discrete group is type~I if and only if it has an
  Abelian normal subgroup of finite index.
\end{thm}

Unfortunately, there are multiply connected configuration spaces of
interest whose fundamental group is not of type~I. For example, the
configuration space for the two-dimensional model describing a charged
quantum particle moving in the magnetic field of $r$ Aharonov-Bohm
fluxes is a plane with $r$ excluded points. It is well known that
$\pi_{1}(\mathbb{R}^{2}\setminus\{ p_{1},\ldots,p_{r}\})$ is the free
group with $r$ generators. However, a freely generated group with two
and more generators is not of type~I. In this case, though, the
situation is not completely lost since a harmonic analysis has been
proved to exist on free groups as well \cite{talamanca}. But we do not
cover this example in the current paper.

The discrete group $\Gamma$ is understood to be equipped with the
counting measure. Let $\mathrm{d}\hat{m}$ be the Plancherel measure on
$\hat{\Gamma}$. Denote by
$\sI_{2}(\sL_{\Lambda})\equiv\sL_{\Lambda}^{\ast}\otimes\sL_{\Lambda}$
Hilbert space formed by Hilbert-Schmidt operators on $\sL_{\Lambda}$
(here $\sL_{\Lambda}^{\ast}$ is the dual space to $\sL_{\Lambda}$).
The Fourier transformation is constructed as a unitary mapping
\begin{equation}
  \label{eq:Fourier_transf}
  \sF:L^{2}(\Gamma)\to\int_{\hat{\Gamma}}^{\oplus}
  \sI_{2}(\sL_{\Lambda})\,\mathrm{d}\hat{m}(\Lambda).
\end{equation}
 Let us list its basic properties \cite{shtern}. If $f_{i}\in L^{2}(\Gamma)$,
$\hat{f}_{i}=\sF[f_{i}]$, $i=1,2$, then
 \[
\sum_{s\in\Gamma}\overline{f_{1}(s)}f_{2}(s)
= \int_{\hat{\Gamma}}\Tr[\hat{f_{1}}(\Lambda)^{\ast}
\hat{f_{2}}(\Lambda)]\,\mathrm{d}\hat{m}(\Lambda).
\]
Furthermore, for $f\in L^{1}(\Gamma)\subset L^{2}(\Gamma)$, one has
\[
\sF[f](\Lambda)=\sum_{s\in\Gamma}f(s)\Lambda(s).
\]
There exists an inversion formula: if $f$ is of the form $f=g\ast{}h$
(the convolution) where $g,h\in L^{1}(\Gamma)$, and $\hat{f}=\sF[f]$
then
\[
f(s)=\int_{\hat{\Gamma}}\Tr[\Lambda(s)^{\ast}\hat{f}(\Lambda)]\,
\mathrm{d}\hat{m}(\Lambda).
\]

Under our restrictions, one does not encounter any problems when
interpreting the above formulas. This is guaranteed in an obvious
manner by the following theorem \cite[Korollar~I]{thoma}.

\begin{thm}[Thoma]
  \label{thm:dimL}
  If $\Gamma$ is a countable discrete group of type~I then
  $\dim\sL_{\Lambda}$ is a bounded function of $\Lambda$ on the dual
  space $\hat{\Gamma}$.
\end{thm}

Consequently, $\sI_{2}(\sL_{\Lambda})$ coincides with the space of all
linear operators on $\sL_{\Lambda}$, and the trace is well defined in
the usual sense. For example, let $\delta_{g}\in L^{2}(\Gamma)$,
$g\in\Gamma$, be defined by $\delta_{g}(s)=\delta_{g,s}$,
$\forall{}s\in\Gamma$. Then $\sF[\delta_{g}](\Lambda)=\Lambda(g)$ and
\[
\|\delta_{g}\|^{2}=\|\sF[\delta_{g}]\|^{2}
= \int_{\hat{\Gamma}}\Tr[\Lambda(g)^{\ast}\Lambda(g)]\,
\mathrm{d}\hat{m}(\Lambda)
= \int_{\hat{\Gamma}}\dim\sL_{\Lambda}\,\mathrm{d}\hat{m}(\Lambda).
\]
 Hence
\[
\int_{\hat{\Gamma}}\dim\sL_{\Lambda}\,\mathrm{d}\hat{m}(\Lambda)=1
\]
and
\begin{equation}
  \label{eq:mGamma_leq1}
  \hat{m}(\hat{\Gamma}) \leq 1.
\end{equation}

Finally, let us note that the Fourier transformation decomposes the
regular representation $\mathcal{R}$ of $\Gamma$ into a direct
integral of irreducible representations. The regular representation
acts on $L^{2}(\Gamma)$ as $\mathcal{R}_{s}=L_{s^{-1}}^{*}$,
$\forall{}s\in\Gamma$, and one has
\[
\forall s\in\Gamma,\,\,\sF\mathcal{R}_{s}\sF^{-1}
=\int_{\hat{\Gamma}}^{\oplus}1\otimes\Lambda(s)\,
\mathrm{d}\hat{m}(\Lambda)
\]
(with the identification
$\sI_{2}(\sL_{\Lambda})\equiv\sL_{\Lambda}^{\ast}\otimes\sL_{\Lambda}$).
This relation means nothing but
\begin{equation}
  \forall s\in\Gamma,\forall f\in L^{2}(\Gamma),\,\,
  \sF[L_{s}^{*}f](\Lambda) = \Lambda(s^{-1})\sF[f](\Lambda).
\label{eq:FL_LF}\end{equation}
In this context, of course, $L_{s}$ stands for the left action of
$\Gamma$ on itself.

\section{The generalized Bloch decomposition}

An application of the harmonic analysis on $\Gamma$ makes it possible
to carry out the first step in the Bloch analysis. This means a
decomposition of the Hilbert space $L^{2}(\tilde{M})$ into a direct
integral jointly with a corresponding decomposition of the Hamiltonian
$H$. Let us describe the procedure in detail. In the notation below,
the variable $y$ usually runs over $\tilde{M}$ while $x$ runs over
$M$. Recall also Remark~\ref{rem:int_f_invar} used repeatedly
throughout this section and, in particular, the meaning of the symbol
$\int_{M}f(y)\,\mathrm{d}\mu(x)$ for a $\Gamma$-invariant function $f$
on $\tilde{M}$.

For $f\in L^{2}(\tilde{M})$ and $y\in\tilde{M}$ set
\[
\forall s\in\Gamma,\, f_{y}(s)=f(s^{-1}\cdot y).
\]
Obviously,
\begin{equation}
  \forall s\in\Gamma,\, f_{s\cdot y}=L_{s^{-1}}^{*}f_{y}
  \label{eq:f_gy_Lgf}
\end{equation}
(here again, $L_s$ stands for the left action of $\Gamma$ on itself).
Thus the norm $\| f_{y}\|$ taken in $L^{2}(\Gamma)$ is a
$\Gamma$-invariant function of $y\in\tilde{M}$ and one easily finds
that
\[
\|f\|^{2} = \int_{M}\| f_{y}\|^{2}\,\mathrm{d}\mu(x).
\]
Hence for almost all $x\in M$ and all $y\in\pi^{-1}(\{ x\})$ one has
$f_{y}\in L^{2}(\Gamma)$. Observe that the tensor product
$\sL_{\Lambda}^{*}\otimes\sH_{\Lambda}$ can be naturally identified
with the Hilbert space of $1\otimes\Lambda$-equivariant
operator-valued functions on $\tilde{M}$ with values in
$\sL_{\Lambda}^{*}\otimes\sL_{\Lambda}\equiv\sI_{2}(\sL_{\Lambda})$.

\begin{define}
  The mapping
  \[
  \Phi:L^{2}(\tilde{M})\to\int_{\hat{\Gamma}}^{\oplus}
  \sL_{\Lambda}^{\ast}\otimes\sH_{\Lambda}\,\mathrm{d}\hat{m}(\Lambda)
  \]
  is defined so that for $f\in L^{2}(\tilde{M})$ and
  $\Lambda\in\hat{\Gamma}$, the component $\Phi[f](\Lambda)$ is a
  measurable operator-valued function on $\tilde{M}$,
  \begin{equation}
    \label{eq:defPhi}
    \Phi[f](\Lambda)\,(y)
    := \sF[f_{y}](\Lambda)\in\sI_{2}(\sL_{\Lambda}).
  \end{equation}
\end{define}

In particular, if $f\in L^{1}(\tilde{M})\cap L^{2}(\tilde{M})$ then
\[
\Phi[f](\Lambda)\,(y)=\sum_{s\in\Gamma}f(s^{-1}\cdot y)\Lambda(s).
\]
From here one can also deduce a simple relation between $\Phi$ and the
mappings $\Phi_{\Lambda}$, $\Lambda\in\hat{\Gamma}$, as introduced in
(\ref{eq:def_PhiL}). For $\varphi\in C_{0}^{\infty}(\tilde{M})$,
$v\in\sL_{\Lambda}$ and $y\in\tilde{M}$,
\begin{equation}
  \Phi[\varphi](\Lambda)(y)v
  = \left(\Phi_{\Lambda}\,\varphi\otimes v\right)\!(y).
  \label{eq:Phi_PhiL}
\end{equation}

\begin{prop}
  $\Phi$ is a well defined unitary mapping.
\end{prop}

\begin{proof}
  (i) According to the above discussion, if $f\in L^{2}(\tilde{M})$
  then for \aall $y\in\tilde{M}$, $\sF[f_{y}](\Lambda)$ is well
  defined for \aall $\Lambda\in\hat{\Gamma}$. By the Fubini theorem,
  for \aall $\Lambda\in\hat{\Gamma}$, the vector-valued function
  $\Phi[f](\Lambda)$ is defined almost everywhere on $\tilde{M}$.
  Moreover, it immediately follows from (\ref{eq:defPhi}),
  (\ref{eq:f_gy_Lgf}) and (\ref{eq:FL_LF}) that $\Phi[f](\Lambda)$ is
  $1\otimes\Lambda$-equivariant.

  (ii) $\Phi$ is an isometry. Indeed (see the defining relation
  (\ref{eq:sprod_psi12}), Remark~\ref{rem:int_f_invar} and
  (\ref{eq:int_Mtild_M})),
  \begin{eqnarray*}
    \|\Phi[f]\|^{2} & = & \int_{\hat{\Gamma}}
    \left(\int_{M}\|{\Phi[f](\Lambda)(y)\|}^{2}\,\mathrm{d}\mu(x)\right)
    \mathrm{d}\hat{m}(\Lambda)\\
    & = & \int_{M}\left(\int_{\hat{\Gamma}}
      \|\sF[f_{y}](\Lambda)\|^{2}\,\mathrm{d}\hat{m}(\Lambda)\right)
    \mathrm{d}\mu(x)\\
    & = & \int_{M}\| f_{y}\|^{2}\,\mathrm{d}\mu(x)\, = \,\| f\|^{2}.
  \end{eqnarray*}

  (iii) $\Phi$ is surjective. Let
  $\psi\in\int_{\hat{\Gamma}}^{\oplus}\sL_{\Lambda}^{*}
  \otimes\sH_{\Lambda}\,\mathrm{d}\hat{m}(\Lambda)$.
  This implies that $\psi(\Lambda)$ is well defined for \aall
  $\Lambda\in\hat{\Gamma}$, and for such $\Lambda$ and for \aall
  $y\in\tilde{M}$, $\psi(\Lambda)(y)\in\sL_{\Lambda}^{*}
  \otimes\sL_{\Lambda}\equiv\sI_{2}(\sL_{\Lambda})$.
  By the Fubini theorem,
  \[
  \|\psi\|^{2}=\int_{M}\left(\int_{\hat{\Gamma}}
    \|\psi(\Lambda)(y)\|^{2}\,\mathrm{d}\hat{m}(\Lambda)\right)
  \mathrm{d}\mu(x).
  \]
  Hence for \aall $x\in M$ and all $y\in\pi^{-1}(\{ x\})$,
  $\psi(\cdot)(y)\in L^{2}(\hat{\Gamma})$.  Since the Fourier
  transform is surjective there exists an essentially unique
  measurable function $\check{\psi}(s,y)$ on $\Gamma\times\tilde{M}$
  such that for \aall $y\in\tilde{M}$,
  $\check{\psi}(\cdot,y)\in{}L^{2}(\Gamma)$ and
  $\sF[\check{\psi}(\cdot,y)](\Lambda)=\psi(\Lambda)(y)$. The
  $\sI_{2}(\sL_{\Lambda})$-valued function $\psi(\Lambda)$ is
  equivariant, i.e., $L_{s}^{*}\psi(\Lambda)=\Lambda(s)\psi(\Lambda)$,
  $\forall s\in\Gamma$. Recalling (\ref{eq:FL_LF}) we have
  \[
  \sF[\check{\psi}(\cdot,s\cdot y)](\Lambda)
  = \Lambda(s)\,\psi(\Lambda)(y)
  = \Lambda(s)\sF[\check{\psi}(\cdot,y)](\Lambda)
  = \sF[L_{s^{-1}}^{*}\check{\psi}(\cdot,y)](\Lambda).
  \]
  From the injectivity of the Fourier transform it follows that
  \[
  \forall s,r\in\Gamma,\textrm{~for~\aall }y\in\tilde{M},\quad
  \check{\psi}(r,s\cdot y)=\check{\psi}(s^{-1}r,y).
  \]
  In particular, letting $r=1$, we have
  \[
  \forall s\in\Gamma,\textrm{~for~\aall }y\in\tilde{M},\quad
  \check{\psi}(s,y)=\check{\psi}(1,s^{-1}\cdot y).
  \]
  Set $f(y)=\check{\psi}(1,y)$. Then one has
  $f_{y}(s)=\check{\psi}(s,y)$ and so $f_{y}\in L^{2}(\Gamma)$ for
  \aall $y\in\tilde{M}$. Moreover,
  $\sF[f_{y}](\Lambda)=\psi(\Lambda)(y)$. From here one easily
  concludes that $f\in L^{2}(\tilde{M})$ and $\Phi[f]=\psi$.
\end{proof}

\begin{prop}
  The decomposition
  \begin{equation}
    \label{eq:H_int_HL}
    \Phi H\Phi^{-1} = \int_{\hat{\Gamma}}^{\oplus}1
    \otimes H_{\Lambda}\,\mathrm{d}\hat{m}(\Lambda)
  \end{equation}
  holds true and, consequently,
  \begin{equation}
    \Phi U(t)\Phi^{-1} = \int_{\hat{\Gamma}}^{\oplus}1
    \otimes U_{\Lambda}(t)\,\mathrm{d}\hat{m}(\Lambda).
    \label{eq:U_int_UL}
  \end{equation}
\end{prop}

\begin{proof}
  Relation (\ref{eq:Phi_PhiL}) between $\Phi$ and $\Phi_{\Lambda}$ and
  equality (\ref{eq:DPhiL_PhiLD}) imply that
  \begin{equation}
    \forall\varphi\in C_{0}^{\infty}(\tilde{M}),\quad
    \Delta_{LB}\Phi[\varphi]=\Phi[\Delta_{LB}\varphi].
    \label{eq:DPhi_PhiD}
  \end{equation}
  Denote by $H^{0}$ the Friedrichs extension of the differential
  operator $-\Delta_{LB}$ with the domain $C_{0}^{\infty}(\tilde{M})$
  in the Hilbert space $L^{2}(\tilde{M})$. Similarly, let
  $H_{\Lambda}^{0}$ be the Friedrichs extension of the differential
  operator $-\Delta_{LB}$ with the domain $\Ran\Phi_{\Lambda}$ in the
  Hilbert space $\sH_{\Lambda}$.  Equality (\ref{eq:DPhi_PhiD})
  implies that for all
  $\varphi_{1},\varphi_{2}\in{}C_{0}^{\infty}(\tilde{M})$,
  \[
  \langle\varphi_{1},-\Delta_{LB}\varphi_{2}\rangle
  = \langle\Phi[\varphi_{1}],\Phi[-\Delta_{LB}\varphi_{2}]\rangle
  = \int_{\hat{\Gamma}}\langle\Phi[\varphi_{1}]
  (\Lambda),1\otimes(-\Delta_{LB})\Phi[\varphi_{2}]
  (\Lambda)\rangle\,\mathrm{d}\hat{m}(\Lambda).
  \]
  Closing the quadratic forms one finds that
  \begin{equation}
    \Phi H^{0}\Phi^{-1} = \int_{\hat{\Gamma}}^{\oplus}1
    \otimes H_{\Lambda}^{0}\,\mathrm{d}\hat{m}(\Lambda).
    \label{eq:H0_int_H0L}
  \end{equation}

  With some abuse of notation we denote by $V$ the multiplication
  operator by the function $V(y)$ in the Hilbert space
  $L^{2}(\tilde{M})$.  According to our assumptions, the function
  $V(y)$ is bounded and so $V$ is a bounded operator. Since the
  function $V(y)$ is $\Gamma$-invariant the corresponding
  multiplication operator can be introduced also in the Hilbert space
  $\sH_{\Lambda}$. In this case the operator will be denoted by the
  symbol $V_{\Lambda}$. The $\Gamma$-invariance of $V(y)$ implies that
  for any $f\in L^{2}(\tilde{M})$, $(Vf)_{y}=V(y)f_{y}$, whence
  \[
  \Phi[Vf](\Lambda)(y)=V(y)\Phi[f](\Lambda)(y).
  \]
  This equality means nothing but
  \begin{equation}
    \Phi V\Phi^{-1}
    = \int_{\hat{\Gamma}}^{\oplus}1\otimes V_{\Lambda}\,
    \mathrm{d}\hat{m}(\Lambda).
    \label{eq:V_int_VL}
  \end{equation}
  Since $H=H^{0}+V$ and $H_{\Lambda}=H_{\Lambda}^{0}+V_{\Lambda}$
  relation (\ref{eq:H_int_HL}) follows from (\ref{eq:H0_int_H0L}) and
  (\ref{eq:V_int_VL}).
\end{proof}

\begin{rem}
  \label{rem:comparison}
  Various forms of decomposition (\ref{eq:H_int_HL}) can be found in
  the literature. Let us make a short comparison to some previous
  works. The generalized Bloch theory has been proposed and used by
  Sunada and collaborators \cite{sunada,adachisunada,adachisunadasy}
  to investigate the character of spectra of $\Gamma$-periodic
  elliptic operators on $\tilde{M}$ under the assumption that the
  quotient $\tilde{M}/\Gamma$ is compact. In particular, in this case
  one can demonstrate the band structure of spectrum, see also
  Ref.~\citen{bruningsunsada} for an asymptotic estimate of the number
  of bands. These ideas have been also applied in
  Refs.~\citen{lledopostI,lledopostII} to construct coverings
  $\tilde{M}\to{}M=\tilde{M}/\Gamma$ such that $M$ is a compact
  Riemannian manifold, $\tilde{M}$ is non-compact and the
  Laplace-Beltrami operator $\Delta_{LB}$ on $\tilde{M}$ has at least
  a prescribed finite number of spectral gaps. The Bloch decomposition
  used in Refs.~\citen{sunada,lledopostII} basically coincides with
  that given in the current paper, only some details of presentation
  somewhat differ (for example, we avoid fixing a fundamental domain
  for the action of $\Gamma$ on $\tilde{M}$). In fact, the situation
  treated in these references is more abstract in the following sense.
  Instead of the Fourier transform (\ref{eq:Fourier_transf}) one can
  consider a unitary mapping
  \begin{displaymath}
    \sF:L^{2}(\Gamma)\to\int_{Z}^{\oplus}\sH(z)\,\mathrm{d}z
  \end{displaymath}
  such that the regular representation $\mathcal{R}_s=L_{s^{-1}}^{*}$
  of $\Gamma$ in $L^2(\Gamma)$ decomposes correspondingly,
  \begin{equation}
    \label{eq:FRsFinverse}
    \forall s\in\Gamma,\,\,\sF\mathcal{R}_s\sF^{-1}
    = \int_{Z}^{\oplus}\mathcal{R}_s(z)\,\mathrm{d}z.
  \end{equation}
  Here $(Z,\mathrm{d}z)$ is supposed to be a separable Hausdorff space
  with a regular Borel measure, and $\mathcal{R}_s(z)$, $z\in{}Z$, is
  a unitary operator in the Hilbert space $\sH(z)$. The unitary
  mapping
  \begin{displaymath}
    \Phi:L^{2}(\tilde{M})\to\int_{Z}^{\oplus}
    \sH_{\mathcal{R}(z)}\,\mathrm{d}z
  \end{displaymath}
  is again defined so that for $f\in L^{2}(\tilde{M})$ and $z\in{}Z$,
  the component $\Phi[f](z)$ is a measurable
  $\mathcal{R}(z)$-equivariant vector-valued function on $\tilde{M}$,
  \begin{displaymath}
    \Phi[f](z)\,(y) := \sF[f_{y}](z)\in\sH(z).
  \end{displaymath}
  Observe that for $h\in{}L^2(\Gamma)$,
  $h=\sum_{s\in\Gamma}h(s)\mathcal{R}_s\delta_e$ ($e\in\Gamma$ is the
  unit element). Furthermore, (\ref{eq:FRsFinverse}) means that
  $\forall{}s\in\Gamma$, $\forall{}h\in{}L^2(\Gamma)$,
  \begin{displaymath}
    \sF[\mathcal{R}_s h](z) = \mathcal{R}_s(z)\sF[h](z)
    \textrm{~~\aevry on~}Z.
  \end{displaymath}
  It follows that for $\varphi\in{}C_0^\infty(\tilde{M})$,
  \begin{equation}
    \label{eq:PhiSunada}
    \Phi[\varphi](z)\,(y) = \sum_{s\in\Gamma}\varphi(s^{-1}\cdot y)
    \mathcal{R}_s(z)\sF[\delta_e](z).
  \end{equation}
  This is the form of $\Phi$ used in Refs.~\citen{sunada,lledopostII}.
  From (\ref{eq:PhiSunada}) one deduces that $\Phi$ decomposes the
  Friedrichs extension of the Laplace-Beltrami operator $\Delta_{LB}$
  on $\tilde{M}$.

  In Ref.~\citen{gruber1} the group $\Gamma$ is supposed to be
  Abelian. On the other hand, one considers therein magnetic
  Hamiltonians with in general non-vanishing magnetic fields. In more
  detail, let $(L,\fh,\nabla)$ be a Hermitian line bundle with
  connection over $M$ and
  $(\tilde{L},\tilde{\fh},\tilde{\nabla})=\pi^\ast(L,\fh,\nabla)$
  where $\pi:\tilde{M}\to{}M$ is the projection. Denote by
  $C^\infty(\tilde{L})$ the vector space of smooth sections in
  $\tilde{L}$. If $\sigma\in{}C^\infty(\tilde{L})$ then
  $\tilde{\nabla}\sigma$ belongs to
  $C^\infty(T^\ast\tilde{M}\otimes\tilde{L})$. The Bochner Laplacian
  $\Delta_B$ as a differential operator acting on
  $C_0^\infty(\tilde{L})$ is unambiguously determined by the equality
  \begin{displaymath}
    \forall\sigma_1,\sigma_2\in C_{0}^{\infty}(\tilde{L}),\quad
    \int_{\tilde{M}}\tilde{\fh}(\sigma_1,-\Delta_B\sigma_2)
    \,\mathrm{d}\tilde{\mu}
    = \int_{\tilde{M}}\tilde{\fg}\otimes\tilde{\fh}
    (\tilde{\nabla}\sigma_1,\tilde{\nabla}\sigma_2)
    \,\mathrm{d}\tilde{\mu}.
  \end{displaymath}
  The magnetic Schr\"odinger operator $H^\nabla$ in the Hilbert space
  $L^2(\tilde{L})$ (the Hilbert space of square integrable sections in
  $\tilde{L}$) is the Friedrichs extension of $-\Delta_B$ with the
  domain $C_{0}^{\infty}(\tilde{L})$. The action of $\Gamma$ on
  $\tilde{M}$ lifts in a canonical way to an isometric linear action
  $\gamma$ of $\Gamma$ on $\tilde{L}$. For $\chi\in\tilde{\Gamma}$ let
  $\sH_{\tilde{L},\chi}$ be the Hilbert space o measurable sections in
  $\tilde{L}$ which satisfy
  \begin{displaymath}
    \forall s\in\Gamma,\quad
    \sigma(s\cdot y) = \chi(s)\,\gamma_s\big(\sigma(y)\big)
    \textrm{~~\aevry on~}\tilde{M}
  \end{displaymath}
  and have a finite norm (see Remark~\ref{rem:int_f_invar})
  \begin{displaymath}
    \int_M\tilde{\fh}(\sigma(y),\sigma(y))\,\mathrm{d}\mu(x)
    < \infty.
  \end{displaymath}
  In Ref.~\citen{gruber1} one constructs in a way very similar to that
  of the current paper a unitary mapping
  \begin{displaymath}
    \Phi:L^{2}(\tilde{L})\to\int_{\hat{\Gamma}}^{\oplus}
    \sH_{\tilde{L},\chi}\,\mathrm{d}\hat{m}(\chi)
  \end{displaymath}
  which decomposes $H^\nabla$.
\end{rem}

\section{The Schwartz kernel theorem}

Equality (\ref{eq:U_int_UL}) represents a way how to express the
evolution operator $U(t)$ in terms of $U_{\Lambda}(t)$,
$\Lambda\in\hat{\Gamma}$.  Our next task is to invert this
relationship. The final formula will concern kernels of operators
rather than directly the operators. Let us recall the fundamental
kernel theorem due to Schwartz (see, for example, Theorem~5.2.1 in
Ref.~\citen{hoermander}).

\begin{thm}[Schwartz]
  Let $X_{i}\subset\mathbb{R}^{n_{i}}$, $i=1,2$, open,
  $\mathcal{K}\in\sD'(X_{1}\times X_{2})$. Then by the equation
  \begin{equation}
    \forall\varphi_{1}\in C_{0}^{\infty}(X_{1}),
    \varphi_{2}\in C_{0}^{\infty}(X_{2}),\quad
    (K\varphi_{1})(\varphi_{2})
    = \mathcal{K}(\varphi_{1}\otimes\varphi_{2})
    \label{schwartz}
  \end{equation}
  there is defined a continuous linear map
  $K:C_{0}^{\infty}(X_{1})\to\sD'(X_{2})$. Conversely, to every such
  continuous linear map $K$ there is one and only one distribution
  $\KK$ such that (\ref{schwartz}) is valid. One calls $\KK$ the
  kernel of $K$.
\end{thm}

\begin{cor}
  To every $B\in\sB(L^{2}(\tilde{M}))$ there exists one and only one
  $\beta\in\sD'(\tilde{M}\times\tilde{M})$ such that
  \[
  \forall\varphi_{1},\varphi_{2}\in C_{0}^{\infty}(\tilde{M}),\quad
  \beta(\overline{\varphi_{1}}\otimes\varphi_{2})
  = \langle\varphi_{1},B\varphi_{2}\rangle.
  \]
  Moreover, the map $B\mapsto\beta$ is injective.
\end{cor}

We call $\beta$ \emph{the kernel of $B$}.

This corollary of the kernel theorem can be also extended to Hilbert
spaces formed by equivariant vector-valued functions. Let us stress
that kernels in this case are operator-valued distributions.

\begin{thm}
  Let $\Lambda$ be an irreducible unitary representation of $\Gamma$
  in a Hilbert space $\sL_\Lambda$. To every $B\in\sB(\sH_{\Lambda})$
  there exists one and only one
  $\beta\in\sD'(\tilde{M}\times\tilde{M})\otimes\sB(\sL_{\Lambda})$
  such that
 \begin{eqnarray*}
   &  & \forall\varphi_{1},\varphi_{2}\in C_{0}^{\infty}(\tilde{M}),
   \forall v_{1},v_{2}\in\sL_{\Lambda},\\
   &  & \langle v_{1},\beta(\overline{\varphi_{1}}
   \otimes\varphi_{2})v_{2}\rangle
   = \langle\Phi_{\Lambda}\,\varphi_{1}
   \otimes v_{1},B\,\Phi_{\Lambda}\,\varphi_{2}
   \otimes v_{2}\rangle.
 \end{eqnarray*}
 The distribution $\beta$ is $\Lambda$-equivariant in the following
 sense
 \begin{equation}
   \forall s\in\Gamma,\quad\beta\circ(L_{s}\otimes1)
   = \Lambda(s)\beta,\,\,\beta\circ(1\otimes L_{s})
   = \beta\Lambda(s^{-1})
   \label{eq:beta_Ls_1}
 \end{equation}
 (here $L_{s}\otimes1$ and $1\otimes L_{s}$ are regarded as
 diffeomorphisms on $\tilde{M}\times\tilde{M}$). Moreover, the map
 $B\mapsto\beta$ is injective.
\end{thm}

\begin{proof}
  From (\ref{eq:prod_PhiL_PhiL}) one can see that
  \[
  \|\Phi_{\Lambda}\varphi\otimes v\|^{2}
  \leq \|v\|^{2}\sum_{s\in\Gamma}|
  \langle\varphi,L_{s}^{*}\varphi\rangle|.
  \]
  Let $K\subset\tilde{M}$ be a compact set. Since the action of
  $\Gamma$ on $\tilde{M}$ is proper there exists a number
  $n_{K}\in\mathbb{N}$ depending only on $K$ such that
  \begin{equation}
    \forall\varphi\in C_{0}^{\infty}(\tilde{M})
    \mbox{ s.t. }\supp\varphi\subset K,\forall v\in\sL_{\Lambda},\quad
    \|\Phi_{\Lambda}\varphi\otimes v\|
    \leq \sqrt{n_{K}}\,\|\varphi\|\| v\|
    \label{eq:norm_PhiL_phi_v}
  \end{equation}
  (here $\|\varphi\|$ is the norm of $\varphi$ in $L^{2}(\tilde{M})$).
  This implies that the linear map
  \newline
  $\Phi_{\Lambda}:C_{0}^{\infty}(\tilde{M})\otimes\sL_{\Lambda}
  \to\sH_{\Lambda}$
  is continuous.

  Fix $v_{1},v_{2}\in\sL_{\Lambda}$ and consider the linear map
  \[
  C_{0}^{\infty}(\tilde{M})\to\sD'(\tilde{M}):\varphi_{1}\mapsto f_{1}
  \]
  defined by
  \[
  \forall\varphi_{2}\in C_{0}^{\infty}(\tilde{M}),\,\,
  f_{1}(\varphi_{2}) = \langle\Phi_{\Lambda}\,\varphi_{1}
  \otimes v_{1},B\,\Phi_{\Lambda}\,\varphi_{2}\otimes v_{2}\rangle.
  \]
  Estimate (\ref{eq:norm_PhiL_phi_v}) implies that this linear map is
  continuous. By the Schwartz kernel theorem, there exists
  $\beta_{v_{1},v_{2}}\in\sD'(\tilde{M}\times\tilde{M}$) such that
  $\beta_{v_{1},v_{2}}(\overline{\varphi_{1}}\otimes\varphi_{2})
  =f_{1}(\varphi_{2})$.
  For $\varphi_{1},\varphi_{2}\in C_{0}^{\infty}(\tilde{M})$, the
  expression
  $\beta_{v_{1},v_{2}}(\overline{\varphi_{1}}\otimes\varphi_{2})$ is
  linear in $v_{2}$, anti-linear in $v_{1}$, and one has
  \[
  |\beta_{v_{1},v_{2}}(\overline{\varphi_{1}}\otimes\varphi_{2})|
  \leq C(\varphi_{1},\varphi_{2})\| B\|\| v_{1}\|\| v_{2}\|
  \]
  where $C(\varphi_{1},\varphi_{2})$ depends only on $\varphi_{1}$,
  $\varphi_{2}$. Hence there exists a unique bounded operator
  $\beta(\overline{\varphi_{1}}\otimes\varphi_{2})\in\sB(\sL_{\Lambda})$
  such that
  \[
  \forall v_{1},v_{2}\in\sL_{\Lambda},\quad
  \langle v_{1},\beta(\overline{\varphi_{1}}
  \otimes\varphi_{2})v_{2}\rangle
  = \beta_{v_{1},v_{2}}(\overline{\varphi_{1}}\otimes\varphi_{2}).
  \]
  Moreover, $\beta(\overline{\varphi_{1}}\otimes\varphi_{2})$ depends
  on $\varphi_{1},\varphi_{2}\in C_{0}^{\infty}(\tilde{M})$
  continuously.  This defines an operator-valued distribution
  $\beta\in\sD'(\tilde{M}\times\tilde{M})\otimes\sB(\sL_{\Lambda})$.

  The diffeomorphism $L_{s}$, $s\in\Gamma$, acting on $\tilde{M}$
  preserves the measure $\tilde{\mu}$. By the definition of
  composition of distributions with diffeomorphisms we have, for any
  $f\in\sD'(\tilde{M})$ and $L_{s}$,
  \[
  \forall\varphi\in C_{0}^{\infty}(\tilde{M}),\quad
  f\circ L_{s}(\varphi)=f(L_{s^{-1}}^{*}\varphi).
  \]
  Recalling (\ref{eq:PhiLLs}) we get
  \begin{eqnarray*}
    \langle v_{1},\beta\circ(L_{s}\otimes1)(\overline{\varphi_{1}}
    \otimes\varphi_{2})v_{2}\rangle
    & = & \langle v_{1},\beta(L_{s^{-1}}^{*}
    \overline{\varphi_{1}}\otimes\varphi_{2})v_{2}\rangle\\
    & = & \langle\Phi_{\Lambda}\circ(L_{s^{-1}}^{*}\otimes1)
    \varphi_{1}\otimes v_{1},B\,\Phi_{\Lambda}\,\varphi_{2}
    \otimes v_{2}\rangle\\
    & = & \langle\Phi_{\Lambda}\,\varphi_{1}
    \otimes\Lambda(s^{-1})v_{1},B\,\Phi_{\Lambda}\,\varphi_{2}
    \otimes v_{2}\rangle\\
    & = & \langle v_{1},\Lambda(s)\beta(\overline{\varphi_{1}}
    \otimes\varphi_{2})v_{2}\rangle.
  \end{eqnarray*}
  This verifies the first relation in (\ref{eq:beta_Ls_1}). The second
  relation in (\ref{eq:beta_Ls_1}) can be verified very similarly.

  The injectivity immediately follows from the fact that
  $\Ran\Phi_{\Lambda}$ is dense in $\sH_{\Lambda}$.
\end{proof}

\section{The Schulman's ansatz}

Everywhere in this section $t$ is a real parameter. We deal with
propagators as distributions introduced as kernels of the
corresponding evolution operators. Let
$\mathcal{K}_{t}\in\sD'(\tilde{M}\times\tilde{M})$ be the kernel of
$U(t)\in\sB(L^{2}(\tilde{M}))$, and let
$\mathcal{K}_{t}^{\Lambda}\in\sD'(\tilde{M}\times\tilde{M})
\otimes\sB(\sL_{\Lambda})$
be the kernel of $U_{\Lambda}(t)\in\sB(\sH_{\Lambda})$. Recall that
$\mathcal{K}_{t}^{\Lambda}$ is $\Lambda$-equivariant which means that
\begin{equation}
  \forall s\in\Gamma,\quad\mathcal{K}_{t}^{\Lambda}(s\cdot y_{1},y_{2})
  = \Lambda(s)\mathcal{K}_{t}^{\Lambda}(y_{1},y_{2}),\,
  \mathcal{K}_{t}^{\Lambda}(y_{1},s\cdot y_{2})
  = \mathcal{K}_{t}^{\Lambda}(y_{1},y_{2})\Lambda(s^{-1}).
  \label{eq:KtLambda_equiv}
\end{equation}

For each $\Lambda\in\hat{\Gamma}$ choose an orthonormal basis
$\{u_{n}^{\Lambda}\}$ in $\sL_{\Lambda}$, and let
$\{\phi_{n}^{\Lambda}\}$ be the dual basis in $\sL_{\Lambda}^{*}$.
First we wish to rewrite the Bloch decomposition of the propagator
(\ref{eq:U_int_UL}) in terms of kernels. Note that if
$A\in\sI_{2}(\sL_{\Lambda})\equiv\sL_{\Lambda}^{*}\otimes\sL_{\Lambda}$
then
\[
A = \sum_{n}\phi_{n}^{\Lambda}\otimes Au_{n}^{\Lambda}.
\]
In particular, if $\varphi\in C_{0}^{\infty}(\tilde{M})$ then
\[
\Phi[\varphi](\Lambda)=\sum_{n}\phi_{n}^{\Lambda}
\otimes\left(\Phi_{\Lambda}\,\varphi\otimes u_{n}^{\Lambda}\right)
\in\sL_{\Lambda}^{*}\otimes\sH_{\Lambda}.
\]
With the aid of this relation one derives that
\begin{eqnarray}
  \langle\Phi[\varphi_{1}](\Lambda),
  \left(1\otimes U_{\Lambda}(t)\right)
  \Phi[\varphi_{2}](\Lambda)\rangle
  & = & \sum_{m}\sum_{n}\langle\phi_{m}^{\Lambda},
  \phi_{n}^{\Lambda}\rangle\langle\Phi_{\Lambda}\,
  \varphi_{1}\otimes u_{m}^{\Lambda},U_{\Lambda}(t)\Phi_{\Lambda}\,
  \varphi_{2}\otimes u_{n}^{\Lambda}\rangle\nonumber \\
  & = & \Tr[\mathcal{K}_{t}^{\Lambda}
  (\overline{\varphi_{1}}\otimes\varphi_{2})].
  \label{eq:Phi_Phi_TrKL}
\end{eqnarray}

\begin{lem}
  For all $\varphi_{1},\varphi_{2}\in C_{0}^{\infty}(\tilde{M})$, the
  function
  $\Lambda\mapsto\Tr[\mathcal{K}_{t}^{\Lambda}
  (\overline{\varphi_{1}}\otimes\varphi_{2})]$
  is integrable on $\hat{\Gamma}$.
\end{lem}

\begin{proof}
  From (\ref{eq:Phi_Phi_TrKL}) one derives the estimate
  \[
  |\Tr[\mathcal{K}_{t}^{\Lambda}
  (\overline{\varphi_{1}}\otimes\varphi_{2})]|
  \leq \|\Phi[\varphi_{1}](\Lambda)\|\|\Phi[\varphi_{2}](\Lambda)\|
  \]
  whence
  \begin{eqnarray*}
    \int_{\hat{\Gamma}}|\Tr[\mathcal{K}_{t}^{\Lambda}
    (\overline{\varphi_{1}}\otimes\varphi_{2})]|\,
    \mathrm{d}\hat{m}(\Lambda)
    & \leq & \left(\int_{\hat{\Gamma}}
      \|\Phi[\varphi_{1}](\Lambda)\|^{2}\,
      \mathrm{d}\hat{m}(\Lambda)\right)^{\!1/2}
    \left(\int_{\hat{\Gamma}}\|\Phi[\varphi_{2}](\Lambda)\|^{2}\,
      \mathrm{d}\hat{m}(\Lambda)\right)^{\!1/2}\\
    & = & \|\varphi_{1}\|\|\varphi_{2}\|\,<\,\infty.
  \end{eqnarray*}
  This concludes the verification.
\end{proof}

\begin{prop}
  The kernel $\mathcal{K}_{t}$, $t\in\R$, decomposes into a direct
  integral in the following sense:
  \begin{equation}
    \forall\varphi_{1},\varphi_{2}\in C_{0}^{\infty}(\tilde{M}),\quad
    \mathcal{K}_{t}(\varphi_{1}\otimes\varphi_{2})
    = \int_{\hat{\Gamma}}\Tr[\mathcal{K}_{t}^{\Lambda}
    (\varphi_{1}\otimes\varphi_{2})]\,\mathrm{d}\hat{m}(\Lambda).
    \label{eq:rel_K_KL}
  \end{equation}
\end{prop}

\begin{proof}
  Applying successively the defining relation for $\mathcal{K}_{t}$,
  equalities (\ref{eq:U_int_UL}) and (\ref{eq:Phi_Phi_TrKL}) one finds
  that
  \begin{eqnarray*}
    \mathcal{K}_{t}(\overline{\varphi_{1}}\otimes\varphi_{2})
    & = & \langle\Phi[\varphi_{1}],\Phi U(t)\Phi^{-1}
    \Phi[\varphi_{2}]\rangle\\
    & = & \int_{\hat{\Gamma}}\langle\Phi[\varphi_{1}](\Lambda),
    \left(1\otimes U_{\Lambda}(t)\right)
    \Phi[\varphi_{2}](\Lambda)\rangle\,\mathrm{d}\hat{m}(\Lambda)\\
    & = & \int_{\hat{\Gamma}}\Tr[\mathcal{K}_{t}^{\Lambda}
    (\overline{\varphi_{1}}\otimes\varphi_{2})]\,
    \mathrm{d}\hat{m}(\Lambda).
  \end{eqnarray*}
  The proof is complete.
\end{proof}

Next we wish to invert relation (\ref{eq:rel_K_KL}). An inverse
relation, which we call here \emph{the Schulman's ansatz}, was derived
in the theoretical physics in the framework of path integration
\cite{schulman1,schulman2} and reads
\[
\mathcal{K}_{t}^{\Lambda}(x,y)
= \sum_{s\in\Gamma}\Lambda(s)\,\mathcal{K}_{t}(s^{-1}\cdot x,y).
\]
Our main goal in the current section is a mathematically rigorous
derivation and interpretation of this formula.

Suppose that $\varphi_{1},\varphi_{2}\in C_{0}^{\infty}(\tilde{M})$
are fixed but otherwise arbitrary. Set
\begin{equation}
  \label{eq:def_Ft}
  F_{t}(s) = \mathcal{K}_{t}\circ(L_{s^{-1}}\otimes1)
  (\varphi_{1}\otimes\varphi_{2})\textrm{~~for~}s\in\Gamma,
\end{equation}
and
\begin{equation}
  \label{eq:def_Gt}
  G_{t}(\Lambda) = \mathcal{K}_{t}^{\Lambda}
  (\varphi_{1}\otimes\varphi_{2})\in\sB(\sL_{\Lambda})
  \textrm{~~for~}\Lambda\in\hat{\Gamma}.
\end{equation} 

\begin{lem}
  \label{thm:Ft_Gt}
  $F_{t}\in L^{2}(\Gamma)$, $G_{t}$ is bounded in the Hilbert-Schmidt
  norm on $\hat{\Gamma}$.
\end{lem}

\begin{proof}
  (i) Since the action of $\Gamma$ is proper one can write any test
  function $\varphi\in C_{0}^{\infty}(\tilde{M})$ as a finite sum,
  $\varphi=\sum_{j=1}^{n}\eta_{j}$, with
  $\eta_{j}\in{}C_{0}^{\infty}(\tilde{M})$, so that
  \[
  \forall j=1,\ldots,n,\forall s\in\Gamma\setminus\{1\},\quad
  \supp\eta_{j}\cap\supp L_{s}^{*}\eta_{j}=\emptyset.
  \]
  This is why one can assume, without loss of generality, that the
  test function $\varphi_{1}$ fulfills
  \begin{equation}
    \forall s\in\Gamma\setminus\{1\},\quad
    \supp\varphi_{1}\cap\supp L_{s}^{*}\varphi_{1} = \emptyset.
    \label{eq:supp_phi1_Ls}
  \end{equation}
  In that case,
  $\{\|\varphi_{1}\|^{-1}\,L_{s}^{*}\varphi_{1}\}_{s\in\Gamma}$ is an
  orthonormal system in $L^{2}(\tilde{M})$. One has
  \[
  F_{t}(s) = \mathcal{K}_{t}(\overline{L_{s}^{*}\varphi_{1}}
  \otimes\varphi_{2})
  = \langle L_{s}^{*}\varphi_{1},U(t)\varphi_{2}\rangle,
  \]
  and so, by the Bessel inequality,
  \[
  \sum_{s\in\Gamma}|F_{t}(s)|^{2} = \|\varphi_{1}\|^{2}
  \sum_{s\in\Gamma}\left|\left\langle\frac{1}{\|\varphi_{1}\|}\,
      L_{s}^{*}\varphi_{1},U(t)\varphi_{2}\right\rangle
  \right|^{2}\leq\|\varphi_{1}\|^{2}\|\varphi_{2}\|^{2}.
  \]

  (ii) Using the defining relation for $\mathcal{K}_{t}^{\Lambda}$ one
  can estimate the Hilbert-Schmidt norm of $G_{t}(\Lambda)$ as follows
  \begin{equation}
    \| G_{t}(\Lambda)\|^{2} = \sum_{m}\sum_{n}
    |\langle u_{m}^{\Lambda},\mathcal{K}_{t}^{\Lambda}
    (\overline{\varphi_{1}}
    \otimes\varphi_{2})u_{n}^{\Lambda}\rangle|^{2}
    \leq \sum_{m}\|\Phi_{\Lambda}\,\varphi_{1}
    \otimes u_{m}^{\Lambda}\|^{2}\sum_{n}\|\Phi_{\Lambda}\,
    \varphi_{2}\otimes u_{n}^{\Lambda}\|^{2}.
    \label{eq:GtL_norm2}
  \end{equation}
  Again, without loss of generality, one can assume that $\varphi_{1}$
  fulfills condition (\ref{eq:supp_phi1_Ls}). Then in the expression
  \[
  \left(\Phi_{\Lambda}\,\varphi_{1}\otimes u_{m}^{\Lambda}\right)(y)
  = \sum_{s\in\Gamma}\varphi_{1}(s^{-1}\cdot
  y)\Lambda(s)u_{m}^{\Lambda}
  \]
  only at most one summand on the RHS does not vanish. It follows that
  \[
  \|\left(\Phi_{\Lambda}\,\varphi_{1}
    \otimes u_{m}^{\Lambda}\right)(y)\|^{2}
  = \sum_{s\in\Gamma}|\varphi_{1}(s^{-1}\cdot y)|^{2}
  \|u_{m}^{\Lambda}\|^{2}
  \]
  and (recalling (\ref{eq:int_Mtild_M}))
  \begin{eqnarray*}
    \|\Phi_{\Lambda}\,\varphi_{1}\otimes u_{m}^{\Lambda}\|^{2}
    & = & \int_{M}\|\left(\Phi_{\Lambda}\,\varphi_{1}
      \otimes u_{m}^{\Lambda}\right)(y)\|^{2}\,\mathrm{d}\mu(x)\\
    & = & \int_{M}\sum_{s\in\Gamma}|\varphi_{1}(s^{-1}\cdot y)|^{2}\,
    \mathrm{d}\mu(x)\,=\,\|\varphi_{1}\|^{2}.
  \end{eqnarray*}
  Hence, if condition (\ref{eq:supp_phi1_Ls}) is true then
  \[
  \sum_{m}\|\Phi_{\Lambda}\,\varphi_{1}\otimes u_{m}^{\Lambda}\|^{2}
  = \dim(\sL_{\Lambda})\,\|\varphi_{1}\|^{2}.
  \]
  In virtue of Theorem~\ref{thm:dimL}, this sum is uniformly bounded
  in $\Lambda$. The other sum in (\ref{eq:GtL_norm2}),
  $\sum_{n}\|\Phi_{\Lambda}\,\varphi_{2}\otimes{}u_{n}^{\Lambda}\|^{2}$,
  can be analogously shown to have the same property.
\end{proof}

\begin{rem}
  Since we know that the total measure $\hat{m}(\hat{\Gamma})$ is
  finite (see (\ref{eq:mGamma_leq1})) Lemma~\ref{thm:Ft_Gt} implies
  that $\|G_{t}(\cdot)\|\in{}L^{1}(\hat{\Gamma})\cap
  L^{2}(\hat{\Gamma})$.
\end{rem}

\begin{prop}
  For all $\varphi_{1},\varphi_{2}\in C_{0}^{\infty}(\tilde{M})$, the
  functions $F_t(s)$ and $G_t(\Lambda)$ defined respectively by
  (\ref{eq:def_Ft}) and (\ref{eq:def_Gt}) satisfy
  \[
  F_{t}=\sF^{-1}[G_{t}].
  \]
\end{prop}

\begin{proof}
  Replacing the test function $\varphi_{1}$ in (\ref{eq:rel_K_KL}) by
  $L_{s}^{\ast}\varphi_{1}$, $s\in\Gamma$, and using
  (\ref{eq:KtLambda_equiv}) one arrives at the equality
  \[
  \mathcal{K}_{t}\circ(L_{s^{-1}}\otimes1)
  (\varphi_{1}\otimes\varphi_{2})
  = \int_{\hat{\Gamma}}\Tr[\Lambda(s)^{\ast}\,
  \mathcal{K}_{t}^{\Lambda}(\varphi_{1}\otimes\varphi_{2})]\,
  \mathrm{d}\hat{m}(\Lambda).
  \]
  In virtue of Lemma~\ref{thm:Ft_Gt}, this means exactly that
  $F_{t}(s)=\sF^{-1}[G_{t}](s)$.
\end{proof}

\begin{cor}
  Conversely,
  \begin{equation}
    G_{t}=\sF[F_{t}].
    \label{schulman}
  \end{equation}
\end{cor}

Rewriting (\ref{schulman}) formally gives
\[
\mathcal{K}_{t}^{\Lambda}(\varphi_{1}\otimes\varphi_{2})
= \sum_{s\in\Gamma}\Lambda(s)\,\mathcal{K}_{t}\circ(L_{s^{-1}}
\otimes1)(\varphi_{1}\otimes\varphi_{2})
\]
which is nothing but the Schulman's ansatz.

\section*{Acknowledgments}

The authors are indebted to Joachim Asch for critical comments to the
manuscript. The authors wish to acknowledge gratefully partial support
from the following grants: grant No. 201/05/0857 of the Grant Agency
of the Czech Republic (P.\v{S}.) and grant No. LC06002 of the Ministry
of Education of the Czech Republic (P.K.).

\end{document}